\documentclass[11pt]{article}

\usepackage[margin=1in]{geometry}
\usepackage{amsmath,amssymb,amsfonts,amsthm,mathtools}
\usepackage{booktabs}
\usepackage{tabularx}
\usepackage{enumitem}
\usepackage{microtype}
\usepackage{hyperref}
\usepackage{tikz}
\usepackage{pgfplots}
\usepackage{authblk}
\usepackage{xcolor}
\usepackage{caption}
\usepackage{bm}
\usetikzlibrary{arrows.meta,positioning,calc,decorations.pathreplacing}
\pgfplotsset{compat=1.18}

\newtheorem{definition}{Definition}
\newtheorem{proposition}{Proposition}
\newtheorem{theorem}{Theorem}
\newtheorem{lemma}{Lemma}
\newtheorem{corollary}{Corollary}
\newtheorem{remark}{Remark}

\newcommand{\E}{\mathbb{E}}

\newcommand{\1}{\mathbf{1}}
\newcommand{\pos}[1]{\left[#1\right]_+}
\newcommand{\Enc}{\mathrm{enc}}
\newcommand{\Pub}{\mathrm{pub}}
\newcolumntype{Y}{>{\raggedright\arraybackslash}X}

\title{\textbf{Reveal, Correct, Then Pay}\\
\large Encrypted Mempools and Perpetual Funding Security}
\author{Benjamin Marsh}
\affil{Sei Labs and University of Portsmouth}
\date{July 2026}

\begin{document}
\maketitle

\begin{abstract}
Encrypted mempools are designed to hide transaction contents until execution order is fixed,
preventing many victim dependent forms of maximal extractable value. This paper studies a different
class of attack in the form of \emph{self-authored state manipulation}, in which the attacker knows
its own transaction and owns a downstream claim on the state that transaction changes.  Perpetual
futures funding is a canonical example.  The funding signal determines a transfer rate, while
receiving side open interest is the transfer base.  In a commit then reveal mempool, an adaptive
corrective transaction cannot enter the already committed batch.  Privacy can therefore create an
economic reaction gap even when cryptographic decryption overhead is negligible. We microfound
correction through executable arbitrage opportunities.  Correctors choose order size against local
price impact and inventory cost, while the protocol information schedule determines which
opportunities are actionable.  The ordering barrier removes ordinary adaptive searchers from the
closed stage.  It therefore yields a closed stage correction rate $\rho_c$ below the adaptive
correction rate $\rho_p$ whenever positive adaptive capacity becomes available after reveal.  The
distortion entering a funding window is multiplied by an explicit response factor
$g(\delta;\rho_c,\rho_p,W)$, where $\delta$ is the first adaptive reaction gap.  We show that $g$ is
strictly increasing in $\delta$ under the derived rate ordering. For an attacker with endowed
receiving notional, the exact expected value of manipulation is quadratic in $m g$, where $m$ is the fraction
of targeting ability retained under privacy.  Encryption is harmful precisely when the loss of
attacker information is smaller than the increase in distortion persistence. We then add hidden tax
base entry.  Transaction privacy can also reduce capitalization of predictable funding into entry
prices, producing a second amplification channel.  The resulting local security index separates
attacker blindness, correction shielding, and capitalization shielding.
\end{abstract}

\section{Introduction}

Encrypted mempools address a real defect of transparent blockchains where a pending transaction can
reveal a user's intent before execution, allowing a validator, builder, or searcher to front-run,
back-run, sandwich, censor, or otherwise condition on the transaction \cite{DaianEtAl2020,QinZhouGervais2022,HeimbachWattenhofer2022,YangEtAl2024,WernerEtAl2022,QinEtAl2021,McLaughlinEtAl2023,QinEtAl2023Imitation}. Threshold encryption systems therefore aim to keep transaction contents secret until inclusion and ordering are fixed \cite{BebelOjha2022,ChoudhuriEtAl2024}.  Recent work has improved the cryptographic and systems
costs of this approach through batched decryption, silent setup, randomized permutation, and direct
integration with high performance BFT protocols \cite{KavousiEtAl2023,BormetEtAl2025,FernandoEtAl2025}.

The standard security story is compelling for \emph{victim dependent} attacks.  A sandwich attack requires the attacker to see another user's trade.  Hide that trade, and the attacker's informational advantage can disappear.  But not all economically relevant attacks have a victim transaction.  In a \emph{self-authored} attack, the adversary knows the transaction that changes application state because it authored that transaction.  The adversary may also hold a claim whose payoff is increasing in the resulting state.  Hiding the transaction then has two opposing effects: it can make the attack less precise, but it can also hide the attack from the agents who would correct it. Perpetual futures funding provides a clean case study.  Perpetuals have no expiry, so periodic funding payments are used to keep the perpetual price near spot \cite{AngerisEtAl2023,AckererEtAl2024, HeEtAl2024}.  Economically, the funding signal sets a transfer rate and receiving side open interest sets the transfer base.  An attacker can hold or acquire receiving notional, distort a mark, index, impact price, or other funding input, and collect the induced transfer.  The attack is not necessarily a reordering attack.  It can be a state manipulation attack whose profitability depends on how quickly arbitrageurs observe and correct the distortion. The key distinction is between \emph{cryptographic latency} and \emph{economic reaction latency}. Suppose an encrypted batch is ordered, then decrypted, then executed.  Even if decryption adds only a few milliseconds, an arbitrageur who learns the plaintext after the order is committed cannot submit a transaction contingent on that plaintext into the same closed batch.  The relevant delay is the time until the first adaptive corrective action can execute.  We call this the \emph{reaction gap}.  It can be one block or one application round even when decryption itself is fast.  This is a cross-layer composition issue, not a failure of the encryption primitive.

This paper makes four contributions.

\begin{enumerate}[leftmargin=1.25cm,label=(\roman*)]
\item \textbf{The ordering barrier.}  We isolate a causal property of commit then reveal designs where
without post-commit admission or precommitted contingent programs, a transaction computed from a
revealed plaintext cannot execute in the committed batch.  This produces a discrete economic reaction
gap independently of decryption computation time.

\item \textbf{Reaction gap manipulation.}  We derive $\rho_c$ and $\rho_p$ from a correction sector
with executable opportunity arrival, local price impact, and inventory cost.  The ordering barrier
implies $\rho_p>\rho_c$ whenever ordinary adaptive correction has positive capacity.  We then derive
the exact amount of a one shot distortion that enters a funding window.  For endowed receiving
notional, this gives a closed form expected attack value and an exact threshold for whether privacy helps or
hurts.

\item \textbf{A two ledger privacy multiplier.}  Funding security depends on both the signal and the
tax base.  We model entry price capitalization of predictable funding and show how hiding the attack
can reduce that capitalization.  A local security index decomposes the effect of encryption into
attacker blindness, correction shielding, and capitalization shielding.

\item \textbf{A decryption aware funding rule.}  We propose the sequence
\emph{reveal, correct, measure, then pay}.  The rule combines a post-reveal correction buffer with
tax base eligibility or time aligned notional accrual.  Under a one-shot hidden impulse, we derive an
explicit buffer sufficient to bound incremental dollar extraction.
\end{enumerate}

Encrypted mempools protect users against important classes of victim dependent MEV.  They do not universally reduce every notion of economic manipulation.  Our result applies when (i) a hidden transaction can move an input to a downstream transfer, (ii) the attacker retains enough information to author that transaction, and (iii) privacy slows or suppresses corrective response or capitalization.  It does not apply in the same form when the application uses a genuinely exogenous, non-manipulable signal or supplies an application native reaction phase.

\section{Related work}
\label{sec:related}

\paragraph{Encrypted mempools.}
Ferveo defines mempool privacy by keeping transactions encrypted until inclusion is finalized and by coupling inclusion to decryption and execution \cite{BebelOjha2022}.  Choudhuri, Garg, Piet, and Policharla formalize batched threshold public key encryption and show how pending transaction privacy can be obtained with communication independent of batch size \cite{ChoudhuriEtAl2025}.  BEAST-MEV combines batched decryption with silent setup \cite{BormetEtAl2025}, while TrX integrates encrypted mempools with a high performance BFT pipeline and reports low proposal to execution overhead \cite{FernandoEtAl2025}.  BlindPerm observes that encryption alone does not eliminate all producer side ordering power and adds randomized permutation after commitment \cite{KavousiEtAl2023}.  These papers primarily study confidentiality, ordering, setup, availability, and systems cost.  We assume those cryptographic goals are met and study the downstream application's transfer rule. Related constructions include FairBlock's minimal overhead front-running prevention \cite{MomeniEtAl2023}, F3B's per-transaction protection \cite{ZhangEtAl2023F3B}, one time setup batched threshold encryption \cite{ChoudhuriEtAl2025}, and epochless batched threshold encryption \cite{BormetEtAl2025BEAT}. These rely on threshold encryption and encryption to the future primitives \cite{CampanelliEtAl2022,DottlingEtAl2023,GargEtAl2024,BonehPartapRotem2024}. A complementary line of work enforces ordering fairness at the consensus layer rather than through encryption \cite{KelkarEtAl2020,Kursawe2020,KelkarDebKannan2022,KelkarEtAl2023Themis,MalkhiSzalachowski2023}.

\paragraph{Economic effects of privacy.}
Rondelet and Kilbourn argue that private mempools must be evaluated through an economic lens, including congestion, information asymmetry, price discovery, and liquidation efficiency \cite{RondeletKilbourn2023}.  In particular, they note that a transparent mempool can allow a liquidation triggered by a price moving transaction to occur in the same block, whereas hidden transaction contents may defer the liquidator's response to a later block.  Our reaction gap model formalizes the analogous mechanism for self-authored state manipulation and derives dollar security conditions.  V0LVER is an example of an application aware design as it combines encrypted orders with an explicit price update and batch clearing mechanism rather than treating encryption as a complete market design \cite{McMenaminDaza2023}. 

\paragraph{Perpetual funding and market design.}
The perpetual futures literature studies no-arbitrage pricing, funding cash flows, and the role of funding in anchoring the perpetual to spot \cite{AngerisEtAl2023,AckererEtAl2024,HeEtAl2024}.  We use only the local transfer structure, a measured premium changes the payment rate, and receiving notional scales the dollar transfer.  Frequent batch auctions show more broadly that a discrete reaction and clearing stage can replace competition on raw speed with competition on price \cite{BudishEtAl2015}. Our proposal is narrower and we claim only that after decryption, an oracle linked application should expose a correction or clearing stage before irreversible measurement and payment. The economics of AMM pricing \cite{FritschEtAl2022} and of batching trades against an AMM \cite{CanidioFritsch2023} are also relevant to how corrective arbitrage prices distortion.

\paragraph{Difference from prior work.}
The cryptographic literature asks whether pending transactions remain confidential until ordering is fixed and whether the resulting protocol is efficient and available. Our paper asks a separate composition question, after a correctly private batch changes application state, does the application immediately pay against that state, or does it first allow the market to correct and capitalize the new information? Automated economic security analysis of smart contracts exists \cite{BabelEtAl2023}, but we are not aware of prior work deriving this reaction gap and tax base interaction for perpetual funding.

\section{Protocol and funding model}
\label{sec:model}

\subsection{Commit, reveal, execute, react}

Consider a batch $B$.  At time $t_c$, the protocol closes admission and fixes the batch's transaction set and execution order.  At time $t_r\ge t_c$, the relevant plaintexts become available.  Transactions are then executed, possibly immediately.  We focus on the common design in which no new transaction can be admitted to $B$ after $t_c$.

\begin{definition}[Adaptive corrective transaction]
Given a revealed transaction $x$, an adaptive corrective transaction is a transaction $c(x)$ whose payload or execution parameters depend nontrivially on the plaintext or execution effect of $x$. Examples include a backrun that restores an AMM price to an external reference, a liquidation triggered by a revealed state change, or an order that offsets a revealed funding signal distortion.
\end{definition}

\begin{lemma}[Ordering barrier]
\label{lem:barrier}
Suppose the admissible transaction set and order of batch $B$ are fixed before the plaintext of $x$ is revealed, and the protocol provides neither post-commit admission nor a precommitted contingent program that implements $c(x)$.  Then an adaptive corrective transaction computed after reveal cannot execute in $B$.  Its earliest execution is in a later batch or application reaction stage.
\end{lemma}

\begin{proof}
At the time $c(x)$ can first be computed, admission to $B$ has closed and the contents of $B$ are fixed. Therefore $c(x)\notin B$.  The conclusion is independent of how quickly the plaintext is decrypted.
\end{proof}

The purpose of the previous lemma is to separate a scheduling fact from a cryptographic benchmark.
Faster threshold decryption can reduce $t_r-t_c$ and materially improve latency
\cite{XiangEtAl2024,FernandoEtAl2025}.  It does not by itself reopen a committed batch. Protocol
native arbitrage, conditional intents, or a reveal then clear phase can remove the barrier.  We return
to these designs in Section~\ref{sec:design}. Let $\delta\geq0$ denote the time from the attack's
state effect to the first execution opportunity for an adaptive correction.  In a stylized transparent
mempool benchmark, searchers can observe a pending
transaction before order closure and submit a contingent backrun, so $\delta=0$.  This is a
benchmark, not a guarantee, since a public mempool may also have a positive effective reaction gap.
All comparisons below remain valid with an arbitrary public value $\delta_{\Pub}$.  The substantive
condition is that the encrypted design raises the effective reaction gap or lowers correction
intensity.

\subsection{Funding as a two ledger transfer}

Let $y\geq0$ be the initial size of the self-authored signal distortion.  Let $Z_t$ be the realized
adversarial component of a funding signal relative to a no-attack path, with $Z_0=y$, and write
$z_t:=\E[Z_t\mid Z_0=y]$.  A funding event measures the average distortion over a window of length
$W>0$.  Only correction before the end of this window matters.  We therefore replace a larger raw
reaction gap by $\min\{\delta,W\}$ and retain the notation $\delta$, so $0\leq\delta\leq W$.  The
expected adversarial component is
\begin{equation}
  \bar z_W
  :=\E\left[
    \left.
    \frac{1}{W}\int_0^W Z_t\,dt
    \right|Z_0=y
  \right]
  =\frac{1}{W}\int_0^W z_t\,dt,
  \label{eq:average-signal}
\end{equation}
and converts it into an expected per unit funding transfer $\tau\bar z_W$, where $\tau>0$ is local
signal to payment pass through.  We use a risk neutral attacker and work on one favorable sign branch
inside any rate cap.  Caps and global position limits bound the scale of the attack but do not change
the local comparative statics.

The attacker may have two types of receiving notional.

\begin{itemize}[leftmargin=0.7cm]
\item $N_0\ge0$ is \emph{endowed notional}: a position held before the hidden transaction is submitted.
It pays no event specific entry premium.
\item $N\ge0$ is \emph{created notional}: a position opened close enough to the event to exploit the
predictable transfer.  It may pay an adverse entry price and other acquisition costs.
\end{itemize}

This is the two ledger structure.  The signal ledger determines the rate, while the notional ledger
determines the dollar base.  A mechanism that protects only one ledger can remain vulnerable through
the other. Let $a y^2$, where $a>0$, be the local all in cost of creating the initial distortion.  Let
$m\in[0,1]$ be the fraction of targeting ability retained by the attacker under the relevant
information regime.  The parameter $m$ absorbs uncertainty about hidden honest flow, ordering,
transaction success, and the sign of interactions with other concealed transactions.  Victim
dependent attacks can have small $m$ under encryption.  A self-authored attack can retain $m$ close
to one because the attacker still knows its own order and any external reference price.

\section{The reaction gap channel}
\label{sec:reaction}

\subsection{Correction from executable arbitrage}

We first derive the two correction rates from corrector behavior and the protocol information
schedule.  Work on the favorable sign branch $Z\geq 0$.  A corrector that can execute an order of
size $x\geq 0$ reduces the distortion from $Z$ to $Z-\ell x$, where $\ell>0$ is local state impact.
The gross value of traversing the resulting linear discrepancy is
\begin{equation}
  \int_0^x (Z-\ell u)\,du
  =Zx-\frac{\ell}{2}x^2.
  \label{eq:corrector-gross-value}
\end{equation}
A corrector of type $i$ has an additional quadratic cost $k_i x^2/2$, where $k_i\geq0$ captures
fees, inventory risk, and limited balance sheet capacity.  Conditional on an actionable correction
opportunity, it therefore solves
\begin{equation}
  \max_{0\leq x\leq Z/\ell}
  \left\{Zx-\frac{\ell+k_i}{2}x^2\right\}.
  \label{eq:corrector-problem}
\end{equation}
The unique solution and the fraction of distortion removed are
\begin{equation}
  x_i^\ast(Z)=\frac{Z}{\ell+k_i},
  \qquad
  \theta_i:=\frac{\ell x_i^\ast(Z)}{Z}
  =\frac{\ell}{\ell+k_i}\in(0,1].
  \label{eq:corrector-response}
\end{equation}

Let $\mathcal I$ be a finite set of corrector types.  Potential opportunities for type $i$ arrive
according to an independent Poisson process with intensity $\nu_i\geq0$.  In stage
$s\in\{c,p\}$, an opportunity is actionable with probability $\alpha_{i,s}\in[0,1]$.
Actionability requires both an executable protocol opportunity and enough information to select the
correct signed order.  A corrector abstains when either condition fails.  Thus the usable
opportunities are a thinned Poisson process with intensity $\nu_i\alpha_{i,s}$.

\begin{proposition}[Correction rate from market primitives]
\label{prop:microfounded-rates}
Suppose the primitives above are constant within stage $s$.  Starting from distortion $Z_0=y$, the
expected residual distortion after time $t$ is
\begin{equation}
  \E[Z_t\mid Z_0=y]=y e^{-\rho_s t},
  \qquad
  \rho_s=\sum_{i\in\mathcal I}
  \nu_i\alpha_{i,s}\frac{\ell}{\ell+k_i}.
  \label{eq:microfounded-rate}
\end{equation}
In particular, $\rho_p\geq\rho_c$ whenever
$\alpha_{i,p}\geq\alpha_{i,c}$ for every type.  The inequality is strict whenever this comparison is
strict for a type with $\nu_i>0$.
\end{proposition}

\begin{proof}
An actionable type $i$ opportunity multiplies the residual distortion by $1-\theta_i$.  Let
$N_{i,s}(t)$ be the number of such opportunities by time $t$.  Poisson thinning gives
$N_{i,s}(t)\sim\operatorname{Pois}(\nu_i\alpha_{i,s}t)$.  Independence and the probability
generating function of a Poisson random variable give
\begin{align}
  \E[Z_t\mid Z_0=y]
  &=y\prod_{i\in\mathcal I}
    \E\left[(1-\theta_i)^{N_{i,s}(t)}\right] \\
  &=y\prod_{i\in\mathcal I}
    \exp\left(-\nu_i\alpha_{i,s}\theta_i t\right)
   =y e^{-\rho_s t}.
\end{align}
The rate comparison follows term by term from \eqref{eq:microfounded-rate}.
\end{proof}

The ordering barrier gives the actionability comparison a direct protocol interpretation.

\begin{corollary}[Ordering induced rate gap]
\label{cor:ordering-rate-gap}
Partition the correction sector into precommitted types $\mathcal C$ and ordinary adaptive types
$\mathcal A$.  A precommitted type has closed stage coverage $q_i\in[0,1]$.  An ordinary adaptive
type has zero closed stage actionability.  Normalize actionability to one once public adaptive
execution resumes.  Then
\begin{align}
  \rho_c
  &=\sum_{i\in\mathcal C}\nu_iq_i\theta_i,
  \label{eq:rho-c-derived}\\
  \rho_p
  &=\sum_{i\in\mathcal C}\nu_i\theta_i
    +\sum_{i\in\mathcal A}\nu_i\theta_i,
  \label{eq:rho-p-derived}\\
  \rho_p-\rho_c
  &=\sum_{i\in\mathcal C}\nu_i(1-q_i)\theta_i
    +\sum_{i\in\mathcal A}\nu_i\theta_i.
  \label{eq:rho-gap-derived}
\end{align}
Hence $\rho_p>\rho_c$ if some precommitted coverage is incomplete or if ordinary adaptive correction
has positive capacity.
\end{corollary}

\begin{proof}
By Lemma~\ref{lem:barrier}, an ordinary adaptive type has $\alpha_{i,c}=0$.  The other actionability
parameters are $\alpha_{i,c}=q_i$ for $i\in\mathcal C$ and $\alpha_{i,p}=1$ for every type.  Substitute
these values into \eqref{eq:microfounded-rate}.
\end{proof}

The parameter $q_i$ is the probability that a fixed program both covers the hidden impulse and can
execute after it.  Normalizing later actionability to one is without loss for the benchmark because a
common later execution probability can be absorbed into $\nu_i$.  More generally,
Proposition~\ref{prop:microfounded-rates} only requires later actionability to dominate closed stage
actionability.

In the homogeneous case, let $\nu_C$ and $\nu_A$ be the total opportunity
intensities of the two sectors, let closed stage coverage be $q$, and let every opportunity remove
fraction $\theta$.  Then
\begin{equation}
  \rho_c=q\nu_C\theta,
  \qquad
  \rho_p=(\nu_C+\nu_A)\theta,
  \qquad
  \frac{\rho_c}{\rho_p}
  =q\frac{\nu_C}{\nu_C+\nu_A}.
  \label{eq:rho-ratio-homogeneous}
\end{equation}
Equation~\eqref{eq:rho-ratio-homogeneous} separates coverage from adaptive capacity.  With no
protocol native or precommitted correction, $\nu_C=0$ and therefore $\rho_c=0$.  Equality
$\rho_c=\rho_p$ requires full closed stage coverage and no correction capacity that becomes newly
actionable after reveal.  A protocol native correction module works by moving capacity from
$\mathcal A$ into $\mathcal C$ and raising its coverage.

The exponential law in Proposition~\ref{prop:microfounded-rates} is exact for expected distortion.
It is also the deterministic fluid limit when individual correction fractions shrink and opportunity
intensities grow while $\nu_i\theta_i$ remains fixed.

\begin{remark}[Fixed correction costs]
The constant rate law is an interior local model.  If type $i$ also pays fixed execution cost
$f_i>0$, its optimized variable surplus is $Z^2/[2(\ell+k_i)]$.  It therefore acts only when
\begin{equation}
  |Z|\geq\sqrt{2f_i(\ell+k_i)}.
  \label{eq:corrector-entry-threshold}
\end{equation}
The local correction rate becomes
\begin{equation}
  \rho_s(Z)
  =\sum_{i\in\mathcal I}
  \nu_i\alpha_{i,s}\frac{\ell}{\ell+k_i}
  \1\left\{|Z|\geq\sqrt{2f_i(\ell+k_i)}\right\}.
  \label{eq:state-dependent-rho}
\end{equation}
Correction can slow as the distortion approaches zero.  Nested actionability still gives
$\rho_p(Z)\geq\rho_c(Z)$ pointwise, but the closed form response factor below need not survive.  The
constant rate model applies on a region where the active set, depth, and inventory limits are fixed.
\end{remark}

\subsection{The induced two stage response}

After the attack changes state, there is an interval $[0,\delta)$ in which a correction newly computed
from the revealed transaction cannot yet execute.  Proposition~\ref{prop:microfounded-rates} gives
rate $\rho_c$ during this closed stage and rate $\rho_p$ once adaptive correction can execute.  For a
one shot impulse $y$, the expected response is
\begin{equation}
  z_t=
  \begin{cases}
    y e^{-\rho_c t}, & 0\le t<\delta,\\[3pt]
    y e^{-\rho_c\delta}e^{-\rho_p(t-\delta)}, & \delta\le t\le W.
  \end{cases}
  \label{eq:two speed-dynamics}
\end{equation}
For $\rho\ge0$ and $s\ge0$, write
\begin{equation}
  A(\rho,s):=
  \begin{cases}
    (1-e^{-\rho s})/\rho, & \rho>0,\\
    s, & \rho=0.
  \end{cases}
  \label{eq:A-def}
\end{equation}
Since the funding transfer is linear in the measured signal, expected attack value depends on the
expected average distortion.  This is $\bar z_W=g(\delta;\rho_c,\rho_p,W)y$, where
\begin{equation}
  g(\delta;\rho_c,\rho_p,W)
  =\frac{A(\rho_c,\delta)+e^{-\rho_c\delta}A(\rho_p,W-\delta)}{W}.
  \label{eq:g-general}
\end{equation}

\begin{proposition}[Reaction gap monotonicity]
\label{prop:g-monotone}
Suppose $0\le\delta<W$ and $\rho_p>0$.  Under the correction sector above,
\begin{equation}
  \frac{\partial g}{\partial\delta}
  =
  \frac{e^{-\rho_c\delta}}{W}
  \left(1-\frac{\rho_c}{\rho_p}\right)
  \left(1-e^{-\rho_p(W-\delta)}\right)
  =
  \frac{e^{-\rho_c\delta}}{W}
  \frac{\rho_p-\rho_c}{\rho_p}
  \left(1-e^{-\rho_p(W-\delta)}\right).
  \label{eq:g-derivative}
\end{equation}
Hence $g$ is weakly increasing in the reaction gap.  It is strictly increasing exactly when some
correction capacity becomes more actionable after reveal, so that $\rho_p>\rho_c$.  It is constant
when the closed stage already exposes the full correction sector and $\rho_p=\rho_c$.
\end{proposition}

\begin{proof}
Differentiate \eqref{eq:g-general}.  Since $\partial_\delta A(\rho_c,\delta)=e^{-\rho_c\delta}$ and $\partial_\delta A(\rho_p,W-\delta)=-e^{-\rho_p(W-\delta)}$, collecting terms gives \eqref{eq:g-derivative}.  Equations~\eqref{eq:rho-c-derived} and \eqref{eq:rho-p-derived} give $\rho_p\geq\rho_c$, with strict inequality under the condition in the statement.
\end{proof}

In a technology matched comparison, the public regime exposes the full correction sector from time
zero, while the encrypted regime delays some of that sector until $\delta$.  Hence
\begin{equation}
  g_{\Pub}=\frac{A(\rho_p,W)}{W},
  \qquad
  g_{\Enc}
  =\frac{A(\rho_c,\delta)
    +e^{-\rho_c\delta}A(\rho_p,W-\delta)}{W}.
  \label{eq:matched-response}
\end{equation}
Proposition~\ref{prop:g-monotone} gives
\begin{equation}
  \frac{g_{\Enc}}{g_{\Pub}}
  =\frac{A(\rho_c,\delta)
    +e^{-\rho_c\delta}A(\rho_p,W-\delta)}{A(\rho_p,W)}
  \geq1,
  \label{eq:matched-amplification}
\end{equation}
with strict inequality when $\delta>0$ and $\rho_p>\rho_c$.

\begin{remark}[Correction technologies that are not nested]
The first equality in \eqref{eq:g-derivative} remains valid when later correction technologies do not
contain the closed stage technologies.  In that extension, $\rho_p<\rho_c$ is possible only because
some correction technology available during the closed stage is disabled later.  It is not produced
by the commit then reveal barrier itself.
\end{remark}

The benchmark with no effective closed stage correction has $\mathcal C=\varnothing$, so
$\rho_c=0$.  Write $\rho:=\rho_p$.  Then
\begin{equation}
  g(\delta;0,\rho,W)
  =\frac{1}{W}\left[\delta+\frac{1-e^{-\rho(W-\delta)}}{\rho}\right].
  \label{eq:g-zero-hidden}
\end{equation}
The transparent same batch reaction benchmark is
\begin{equation}
  g(0;0,\rho,W)=\frac{1-e^{-\rho W}}{\rho W}.
  \label{eq:g-public}
\end{equation}
Thus the reaction gap amplification factor is
\begin{equation}
  R_g(\delta)
  :=\frac{g(\delta;0,\rho,W)}{g(0;0,\rho,W)}
  =\frac{\rho\delta+1-e^{-\rho(W-\delta)}}{1-e^{-\rho W}}.
  \label{eq:reaction-amplification}
\end{equation}
If $\rho W\to\infty$ and $\rho(W-\delta)\to\infty$, then
\begin{equation}
  R_g(\delta)=1+\rho\delta+o(1).
  \label{eq:reaction-asymptotic}
\end{equation}
This comparative static is easy to miss, a fixed one block reaction gap can matter most in markets where transparent corrective arbitrage would otherwise be fastest.

\begin{figure}[t]
\centering
\begin{tikzpicture}
\begin{axis}[
  width=0.78\textwidth,
  height=5.4cm,
  xmin=0,xmax=0.25,
  ymin=1,ymax=6.2,
  xlabel={reaction gap $\delta/W$},
  ylabel={$g(\delta)/g(0)$},
  legend pos=north west,
  axis lines=left,
  samples=160,
  scaled ticks=false,
  xtick={0,0.05,0.10,0.15,0.20,0.25},
  xticklabels={0,0.05,0.10,0.15,0.20,0.25},
  grid=major,
  grid style={dotted}
]
\addplot[thick,domain=0:0.25] {(x+(1-exp(-2*(1-x)))/2)/((1-exp(-2))/2)};
\addlegendentry{$\rho W=2$}
\addplot[thick,dashed,domain=0:0.25] {(x+(1-exp(-8*(1-x)))/8)/((1-exp(-8))/8)};
\addlegendentry{$\rho W=8$}
\addplot[thick,dashdotted,domain=0:0.25] {(x+(1-exp(-20*(1-x)))/20)/((1-exp(-20))/20)};
\addlegendentry{$\rho W=20$}
\end{axis}
\end{tikzpicture}
\caption{Reaction gap amplification with no correction during the closed stage.  A short gap has a larger proportional effect when public correction would otherwise be rapid.}
\label{fig:reaction-amplification}
\end{figure}
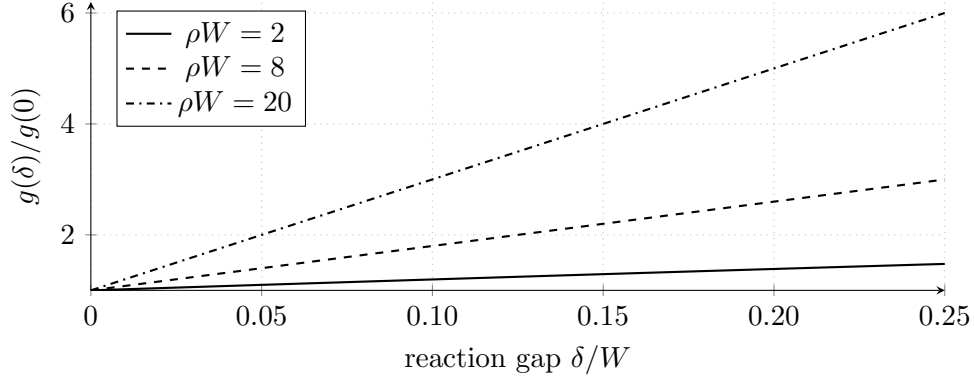

\subsection{Endowed notional and the information threshold}

An attacker with endowed receiving notional $N_0$ has expected incremental payoff
\begin{equation}
  \Pi_0(y)=mN_0\tau g\,y-a y^2,
  \label{eq:endowed-payoff}
\end{equation}
where $g=g(\delta;\rho_c,\rho_p,W)$.

\begin{proposition}[Endowed notional value]
\label{prop:endowed-value}
The optimal one shot distortion and its expected dollar value are
\begin{equation}
  y^\ast=\frac{mN_0\tau g}{2a},
  \qquad
  V_0=\frac{m^2N_0^2\tau^2g^2}{4a}.
  \label{eq:endowed-solution}
\end{equation}
\end{proposition}

\begin{proof}
Maximize the concave quadratic \eqref{eq:endowed-payoff} in $y$.
\end{proof}

Let $j\in\{\Pub,\Enc\}$ index public and encrypted regimes.  Allow the local distortion cost and pass
through to differ across regimes.  Assume the public benchmark has $V_{0,\Pub}>0$, so the ratios below
are defined.  Proposition~\ref{prop:endowed-value} gives
\begin{equation}
  \frac{V_{0,\Enc}}{V_{0,\Pub}}
  =
  \frac{a_{\Pub}}{a_{\Enc}}
  \left(
    \frac{m_{\Enc}\tau_{\Enc}g_{\Enc}}
         {m_{\Pub}\tau_{\Pub}g_{\Pub}}
  \right)^2.
  \label{eq:endowed-ratio-general}
\end{equation}
With unchanged execution cost and application pass through,
\begin{equation}
  V_{0,\Enc}>V_{0,\Pub}
  \quad\Longleftrightarrow\quad
  \frac{m_{\Enc}}{m_{\Pub}}
  >
  \frac{g_{\Pub}}{g_{\Enc}}.
  \label{eq:information-threshold}
\end{equation}
There is therefore no unconditional ``privacy improves manipulation security'' theorem.  Privacy helps when the attacker's lost targeting ability dominates the persistence gained from weaker correction. For a self-authored attack, $m_{\Enc}/m_{\Pub}$ may be close to one, in which case any increase in $g$ raises attack value.  For a victim-dependent attack, the information ratio can be near zero, in which case privacy is strongly protective.

As a synthetic illustration, take $\rho W=8$, no closed stage correction, and $\delta/W=0.1$.  Then $R_g\approx1.80$, so attack value rises by a factor of about $3.24$ when targeting is unchanged. Encryption still raises value whenever the attacker retains more than $1/1.80\approx56\%$ of its public mempool targeting ability.  The numbers are not a venue calibration, they show that a discrete reaction opportunity can dominate modest information loss.

\section{The tax base channel}
\label{sec:taxbase}

Reaction speed prices only the signal ledger. Funding has a second ledger, who owns the receiving notional when the payment is made.  A timestamp based eligibility rule can allow an attacker to enter in the concealed batch after the transfer is largely predictable.  Public observability can make that entry expensive because competing traders bid the predictable funding into the entry price.  Encryption can suppress that capitalization.

\subsection{Visibility and capitalization}

Let
\begin{equation}
  p=q_0y,
  \qquad
  q_0:=m\tau g,
  \label{eq:predictable-payment}
\end{equation}
be the predictable expected incremental funding payment per unit of receiving notional generated by distortion $y$.  Let $e$ be the adverse entry premium paid by a trader entering the receiving side before the funding event.  Introduce an effective visibility parameter $v\in[0,1]$.  The value $v=1$ means that competing traders can condition their entry on the pending manipulation, $v=0$ means that no transaction specific entry can be submitted before the relevant batch closes.

Competing sniping demand and passive entry supply are
\begin{equation}
  D(e;p)=v\kappa_M(p-e)_+,
  \qquad
  S(e)=\kappa_L e,
  \label{eq:entry-market}
\end{equation}
where $\kappa_M\ge0$ measures responsive capital and $\kappa_L>0$ is entry depth.

\begin{proposition}[Visibility dependent capitalization]
\label{prop:capitalization}
For $p>0$, the entry market clears at
\begin{equation}
  e^\ast=\zeta(v)p,
  \qquad
  \zeta(v)=\frac{v\kappa_M}{v\kappa_M+\kappa_L}.
  \label{eq:zeta}
\end{equation}
The capitalization coefficient is increasing in visibility and responsive competing capital, and
decreasing in entry depth.  A newly entered receiver retains only $(1-\zeta)p$ before residual costs.
\end{proposition}

\begin{proof}
On the positive branch, market clearing gives $v\kappa_M(p-e)=\kappa_Le$.  Solving for $e$ yields
\eqref{eq:zeta}.
\end{proof}

The parameter $v$ is not synonymous with raw cryptographic leakage.  It is the fraction of the transaction specific opportunity that can be converted into a competing order before the eligibility
cutoff.  A decrypt then clear mechanism can have high effective $v$ even if the pre-commit mempool is perfectly private. A serial order then reveal execute mechanism can have $v\approx0$ for same batch
entry even when decryption itself is immediate.

\subsection{The two ledger security index}

Let $N$ be notional created for the event.  Its local acquisition, carry, adverse self-impact, and unwind cost is summarized by $bN^2$, $b>0$.  Endowed notional $N_0$ receives the full predictable payment, while created notional receives the residual after capitalization.  Define
\begin{equation}
  q:=(1-\zeta)q_0=(1-\zeta)m\tau g.
  \label{eq:q-created}
\end{equation}
The expected local payoff is
\begin{equation}
  \Pi(N,y)=q_0N_0y+qNy-a y^2-bN^2.
  \label{eq:two-ledger-payoff}
\end{equation}

\begin{theorem}[Two ledger amplification]
\label{thm:two-ledger}
Define the local tax base security index
\begin{equation}
  \Gamma:=\frac{q^2}{4ab}
  =
  \frac{(1-\zeta)^2m^2\tau^2g^2}{4ab}.
  \label{eq:Gamma}
\end{equation}
If $\Gamma<1$, the unique optimum is
\begin{equation}
  y^\ast=\frac{2bq_0N_0}{4ab-q^2},
  \qquad
  N^\ast=\frac{qq_0N_0}{4ab-q^2},
  \label{eq:two-ledger-controls}
\end{equation}
with value
\begin{equation}
  V=\frac{q_0^2N_0^2}{4a(1-\Gamma)}.
  \label{eq:two-ledger-value}
\end{equation}
If $\Gamma>1$, the quadratic model has a positive homogeneous attack direction and is locally
scalable until a cap, budget, position limit, or global impact curve binds.  At $\Gamma=1$, the
homogeneous quadratic vanishes along a nonzero ray.  If $q_0N_0>0$, endowed notional creates
unbounded linear gain along that ray.  If $q_0N_0=0$, payoff is zero along the ray rather than
unbounded.
\end{theorem}

\begin{proof}
The Hessian of \eqref{eq:two-ledger-payoff} in $(y,N)$ is negative definite exactly when
$4ab-q^2>0$.  Under that condition, the first order conditions are $2ay-qN=q_0N_0$ and $2bN=qy$.
Solving gives \eqref{eq:two-ledger-controls}, and substitution gives \eqref{eq:two-ledger-value}.  If
$q^2>4ab$, the homogeneous quadratic form has a positive eigenvalue.  At the boundary, evaluate the
linear endowed notional term along the zero curvature ray.  It is positive exactly when $q_0N_0>0$.
\end{proof}

The factor $(1-\Gamma)^{-1}$ is a tax base amplification multiplier.  Endowed notional creates a finite seed value even when created notional alone is not locally profitable.  Cheap hidden entry then scales up both the optimal distortion and the receiving base.  As $\Gamma$ approaches one, small changes in reaction speed or capitalization have large effects.

For two regimes with the same $a,b,$ and $\tau$, assume $m_{\Pub}>0$ and
$\zeta_{\Pub},\zeta_{\Enc}<1$.  The security index ratio is then
\begin{equation}
  \frac{\Gamma_{\Enc}}{\Gamma_{\Pub}}
  =
  \underbrace{\left(\frac{m_{\Enc}}{m_{\Pub}}\right)^2}_{\text{attacker blindness}}
  \underbrace{\left(\frac{g_{\Enc}}{g_{\Pub}}\right)^2}_{\text{correction shielding}}
  \underbrace{\left(
    \frac{1-\zeta_{\Enc}}{1-\zeta_{\Pub}}
  \right)^2}_{\text{capitalization shielding}}.
  \label{eq:Gamma-decomposition}
\end{equation}
When both regimes satisfy $\Gamma<1$ and $V_{\Pub}>0$, their finite expected values satisfy
\begin{equation}
  \frac{V_{\Enc}}{V_{\Pub}}
  =
  \left(
    \frac{m_{\Enc}g_{\Enc}}{m_{\Pub}g_{\Pub}}
  \right)^2
  \frac{1-\Gamma_{\Pub}}{1-\Gamma_{\Enc}}.
  \label{eq:value-decomposition}
\end{equation}
Equation~\eqref{eq:Gamma-decomposition} is the central composition result.  Encryption can reduce the
first term while increasing the second and third.  For local tax base scalability, encryption is harmful exactly when
\begin{equation}
  \frac{m_{\Enc}}{m_{\Pub}}
  >
  \frac{g_{\Pub}}{g_{\Enc}}
  \frac{1-\zeta_{\Pub}}{1-\zeta_{\Enc}},
  \label{eq:Gamma-threshold}
\end{equation}
under unchanged costs and pass through.

The capitalization channel can make this threshold much lower than the endowed notional threshold \eqref{eq:information-threshold}.  For example, if public visibility capitalizes $80\%$ of predictable funding, encrypted same batch entry capitalizes none, and the reaction gap factor is $1.8$, then $\Gamma_{\Enc}>\Gamma_{\Pub}$ whenever the attacker retains more than roughly $11\%$ of its public information effectiveness.  Again, this is a model illustration rather than a venue estimate.  It shows why a mechanism that hides both the signal trade and the recipient's entry can differ sharply from one that merely hides a victim's order.

\begin{figure}[t]
\centering
\begin{tikzpicture}
\begin{axis}[
  width=0.78\textwidth,
  height=5.4cm,
  xmin=0,xmax=0.22,
  ymin=0,ymax=1.02,
  xlabel={reaction gap $\delta/W$},
  ylabel={critical $m_{\Enc}/m_{\Pub}$},
  legend pos=north east,
  axis lines=left,
  samples=160,
  scaled ticks=false,
  xtick={0,0.05,0.10,0.15,0.20},
  xticklabels={0,0.05,0.10,0.15,0.20},
  grid=major,
  grid style={dotted}
]
\addplot[thick,domain=0:0.22] {((1-exp(-8))/8)/(x+(1-exp(-8*(1-x)))/8)};
\addlegendentry{$\zeta_{\Enc}=0.8$}
\addplot[thick,dashed,domain=0:0.22] {((1-exp(-8))/8)/(x+(1-exp(-8*(1-x)))/8)*(0.2/0.6)};
\addlegendentry{$\zeta_{\Enc}=0.4$}
\addplot[thick,dashdotted,domain=0:0.22] {((1-exp(-8))/8)/(x+(1-exp(-8*(1-x)))/8)*0.2};
\addlegendentry{$\zeta_{\Enc}=0$}
\end{axis}
\end{tikzpicture}
\caption{Critical attacker information retention for $\Gamma_{\Enc}>\Gamma_{\Pub}$, with
$\rho W=8$, no closed stage correction, and $\zeta_{\Pub}=0.8$.  Reduced capitalization can make the
tax base channel dominate substantial attacker blindness.}
\label{fig:critical-information}
\end{figure}
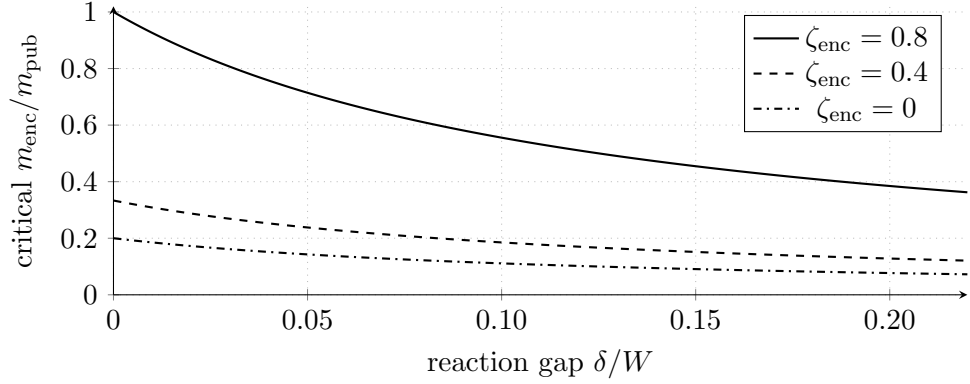

\section{Reveal, correct, measure, then pay}
\label{sec:design}

The result is not an argument against encrypted mempools.  It is an argument against composing a commit then reveal ordering layer with an application that immediately treats the first post-reveal state as a payment oracle.  The application should insert an economically meaningful stage between reveal and payment.

\subsection{A decryption aware rule}

Consider the following sequence.

\begin{enumerate}[leftmargin=1.25cm,label=\textbf{Step \arabic*.}]
\item \textbf{Order and reveal.}  Commit the encrypted batch, reveal its transactions, and execute or
simulate the resulting state transition.
\item \textbf{Correct.}  Permit a public reaction period of length $L$, a reaction auction, or a
protocol native rebalancing step.  The reaction mechanism must be able to condition on the revealed
state.
\item \textbf{Measure.}  Begin the funding signal window only after the correction stage, or exclude
the reveal-induced transient from the window.
\item \textbf{Pay on an aligned tax base.}  Positions opened in the concealed batch do not receive the imminent event.  Alternatively, receiving notional accrues continuously over the same post-reveal window used to measure the signal rather than at one final timestamp.
\end{enumerate}

The rule can be implemented in several ways.  A chain can add a reveal then clear subround.  An
application can delay oracle sampling by a fixed or state dependent number of blocks.  A protocol native arbitrage module can rebalance against an external price before funding is finalized.  Or the funding contract can determine event eligibility from a pre-batch snapshot and accrue newly opened positions only from the next event.  These mechanisms differ in liveness and market-quality costs, but they target the two channels identified above.

\subsection{An explicit correction buffer}

Suppose that, at reveal, the residual hidden order distortion is $y$ and that its conditional mean
satisfies $z_t\leq y e^{-\rho t}$ after reveal for some $\rho>0$.  The application waits $L$ and then
measures over a window of length $W$.  In the absence of further hidden intervention, the expected
average measured distortion satisfies
\begin{equation}
  \bar z_{L,W}
  \leq y\,h(L;\rho,W),
  \qquad
  h(L;\rho,W)
  :=e^{-\rho L}\frac{1-e^{-\rho W}}{\rho W}.
  \label{eq:buffer-response}
\end{equation}

\begin{theorem}[Decryption aware dollar bound]
\label{thm:buffer-bound}
Assume:
\begin{enumerate}[leftmargin=1.15cm,label=(\roman*)]
\item the residual hidden order distortion at reveal has local creation cost $ay^2$;
\item the conditional mean distortion after reveal decays at exponential rate at least $\rho$;
\item signal to payment pass-through is at most $\tau$;
\item endowed receiving notional is at most $\overline N$;
\item positions created in the concealed batch are ineligible for the current event.
\end{enumerate}
Then the maximum expected incremental funding value generated by the concealed batch residual distortion is bounded by
\begin{equation}
  \varepsilon_{\mathrm{priv}}(L)
  =
  \frac{m^2\overline N^2\tau^2}{4a}
  e^{-2\rho L}
  \left(\frac{1-e^{-\rho W}}{\rho W}\right)^2.
  \label{eq:privacy-bound}
\end{equation}
For every $\varepsilon>0$, to guarantee
$\varepsilon_{\mathrm{priv}}(L)\le\varepsilon$, it is sufficient to choose
\begin{equation}
  L\ge
  \frac{1}{2\rho}
  \pos{
    \log\left[
      \frac{m^2\overline N^2\tau^2}{4a\varepsilon}
      \left(\frac{1-e^{-\rho W}}{\rho W}\right)^2
    \right]
  }.
  \label{eq:buffer-rule}
\end{equation}
\end{theorem}

\begin{proof}
Under the stated eligibility rule, concealed batch entry does not contribute created notional to the
current event, so the privacy specific tax base amplification term is absent.
Proposition~\ref{prop:endowed-value} applies with response factor at most $h(L;\rho,W)$ and
$N_0\le\overline N$, yielding \eqref{eq:privacy-bound}.  Solving the inequality for $L$ gives
\eqref{eq:buffer-rule}.
\end{proof}

The theorem bounds the incremental transfer generated by the concealed batch residual.  It does not claim that all post-reveal manipulation disappears.  An attacker can continue trading publicly after reveal, those trades belong in the ordinary oracle manipulation and dollar security analysis.  The role of the buffer is to ensure that the protocol does not grant a special uncorrected payment window solely because the attack was concealed during ordering.

\subsection{Why both controls are needed}

A correction buffer and an eligibility rule address different ledgers.

\begin{table}[t]
\centering
\small
\begin{tabularx}{\textwidth}{>{\bfseries\raggedright\arraybackslash}p{0.25\textwidth} Y Y}
\toprule
Design & Signal ledger & Tax base ledger \\
\midrule
Immediate post-reveal measurement
& Preserves the reaction gap distortion.
& Concealed entrants may receive the event. \\
\addlinespace
Correction buffer only
& Reduces $g$ or $h(L)$.
& A final timestamp rule can still make entry cheap. \\
\addlinespace
Eligibility lag only
& Removes same batch created notional from the event.
& Endowed notional still profits from an uncorrected signal. \\
\addlinespace
Reveal correct measure pay
& Lets adaptive arbitrage act before measurement.
& Excludes concealed entrants or accrues notional over the public measurement window. \\
\addlinespace
Exogenous non-manipulable oracle
& Removes or sharply weakens the self-authored signal channel.
& Tax base timing still matters for predictable transfers. \\
\bottomrule
\end{tabularx}
\caption{The signal and tax base controls are complements, not substitutes.}
\label{tab:controls}
\end{table}

An eligibility lag alone sets the privacy specific created notional coefficient $q$ to zero but leaves $q_0$ intact for endowed positions.  A signal buffer alone reduces both $q_0$ and $q$ but can leave $\zeta$ near zero for a timestamp entrant.  The combined rule is therefore the natural application level composition of transaction privacy with a state contingent transfer.

\section{Numerical illustrations and design implications}
\label{sec:numerical}

Table~\ref{tab:reaction-numbers} reports dimensionless reaction gap effects with no closed stage correction. The first ratio is the measured distortion multiplier $R_g$; the second is the endowed notional value multiplier when $m$ is unchanged; the final column is the fraction of public targeting ability the attacker must retain for encryption to increase endowed attack value.

\begin{table}[h]
\centering
\small
\begin{tabular}{c c c c c}
\toprule
$\rho W$ & $\delta/W$ & $R_g$ & $R_g^2$ & critical $m_{\Enc}/m_{\Pub}$ \\
\midrule
2  & 0.10 & 1.197 & 1.432 & 0.836 \\
8  & 0.05 & 1.400 & 1.960 & 0.714 \\
8  & 0.10 & 1.800 & 3.239 & 0.556 \\
20 & 0.05 & 2.000 & 4.000 & 0.500 \\
20 & 0.10 & 3.000 & 9.000 & 0.333 \\
\bottomrule
\end{tabular}
\caption{Synthetic reaction gap comparative statics for $\rho_c=0$.}
\label{tab:reaction-numbers}
\end{table}

Several design implications follow.

\paragraph{Fast decryption is necessary but not sufficient.}
Low cryptographic overhead improves user latency and reduces one component of wall clock delay.  But if the application closes adaptive admission before reveal, the remaining one round ordering barrier can still be material. A systems benchmark should therefore report not only proposal to execution latency but also the earliest stage at which a transaction contingent on the revealed state can execute.

\paragraph{Protocol native correction is especially valuable under privacy.}
A precommitted arbitrage module, solver auction, or state rebalancing hook can move correction
capacity from $\mathcal A$ into $\mathcal C$.  Moving a fully covered type $i$ in this way raises
$\rho_c$ by $\nu_i\theta_i$ while leaving $\rho_p$ unchanged.  It therefore narrows the correction
deficit during the reaction gap.  Adding a new correction technology that remains active in both
stages raises both rates.  Reducing $\delta$ attacks the timing margin directly.  Which agents can
supply these forms of correction is itself shaped by order flow and searcher competition in block
building \cite{GuptaPaiResnick2023,MamageishviliEtAl2024}.

\paragraph{Randomized permutation mainly acts through $m$.}
Permutation can make a manipulator less certain about relative order and interactions with other hidden transactions, lowering targeting retention.  This can be strongly protective.  It does not automatically restore an adaptive correction opportunity after the plaintext is known.  In the model, permutation primarily lowers $m$; a reveal then clear stage primarily lowers $g$ and raises effective visibility $v$. The controls are complementary.

\paragraph{Long funding windows do not make the issue disappear mechanically.}
A small $\delta/W$ can still matter when $\rho W$ is large.  The relevant dimensionless quantity in the fast correction regime is approximately $\rho\delta$, not merely the fraction of the funding window lost. That said, if the signal is derived from a deep external oracle and the hidden venue transaction barely moves it, the effective $a$ is large and the channel is weak.

\paragraph{Dollar security requires concentration limits.}
All attack values scale with $N_0^2$.  An application can reduce the required correction buffer through account, coalition, or market wide limits on receiving notional.  A rate cap is not a complete substitute, it bounds basis points, while the dollar transfer still scales with the attackable tax base.

\section{Conclusion}
Encrypted mempools remove information from would be attackers, but they also remove information from
correctors and competing entrants.  For victim dependent MEV, the first effect can dominate.  For
self-authored state manipulation, the attacker may retain most of the information it needs while the
market loses the opportunity to respond. Perpetual funding makes the composition failure visible.  A
hidden transaction can affect the signal that sets a transfer rate, while endowed or newly created
open interest supplies the receiving base.

The correction sector makes the missing market response explicit.  The difference
$\rho_p-\rho_c$ is the weighted correction capacity that gains actionable information or execution
access after reveal.  The reaction gap factor $g$ prices the delay before that capacity can act.  The
capitalization coefficient $\zeta$ prices how much predictable funding is bid into entry.  Their
combination yields a security decomposition into attacker blindness, correction shielding, and
capitalization shielding.

The practical rule is simple.  Do not reveal and immediately pay against the revealed state. Reveal,
permit correction or clearing, measure after the transient, and align position eligibility with the
public measurement window.  Transaction privacy and market integrity are compatible, but only when
the application is designed for the information schedule that encryption creates.

\end{document}